\documentclass[11pt]{article}
\usepackage{float}

\usepackage{color}

\usepackage{algorithm}
\usepackage{algpseudocode}

\usepackage{tabularx}
\usepackage{fullpage}
\newfloat{algo}{htbp}{algo}
\floatname{algo}{Algorithm}
\usepackage{url}

\usepackage{times}

\usepackage{graphicx}
\usepackage{multicol}

\usepackage{amssymb}
\usepackage{amsmath,wasysym}
\allowdisplaybreaks[2]          % but allow some broken displays

\usepackage{amssymb} %%% symbols %%%

\usepackage{amssymb} %%% symbols %%%

\newtheorem{defn}{Definition}[section]
\newtheorem{lemma}[defn]{Lemma}
\newtheorem{cor}[defn]{Corollary}
\newtheorem{thm}[defn]{Theorem}
\newtheorem{observation}{Observation}

\newcommand*{\qed}{\hfill\ensuremath{\square}}%
\newenvironment{proof}{\vspace{1ex}\noindent{\bf Proof:}\hspace{0.5em}}
	{\hfill\qed\vspace{1em}}

\newcommand{\Write}{\ensuremath{\textsc{write}}}
\newcommand{\Read}{\ensuremath{\textsc{read}}}

\newcommand{\CAS}{\ensuremath{\textit{CAS}}}

\newcommand{\PntCont}{\ensuremath{\textrm{PntCont}}}
\definecolor{light-gray}{gray}{0.5}

\algdef{SE}[OPERATION]{Operation}{EndOperation}[1]{\textbf{operation}\
#1\ }{\algorithmicend\ \textbf{}}%
\algtext*{EndEvent}

\title{Space Bounds for Reliable Multi-Writer Data Store:\\ Inherent
  Cost of Read/Write Primitives}
\author{
  Gregory Chockler\\
  Royal Holloway, University of London\\
  \texttt{Gregory.Chockler@rhul.ac.uk}
  \and
  Dan Dobre\\
  NEC Labs Europe\\
  \texttt{dan.dobre@neclab.eu}
  \and
  Alexander Shraer\\
  Google, Inc.\\
  \texttt{shralex@google.com}
  \and
  Alexander Spiegelman\\
  Technion\\
  \texttt{sashas@tx.technion.ac.il}
}

\date{}

\begin{document}

\maketitle
\thispagestyle{empty}

\begin{abstract}
  Reliable storage emulations from fault-prone components have
  established themselves as an algorithmic foundation of modern
  storage services and applications. Most existing reliable storage
  emulations are built from storage services supporting arbitrary
  read-modify-write primitives. Since such primitives are not
  typically exposed by pre-existing or off-the-shelf components (such
  as cloud storage services or network-attached disks) it is natural
  to ask if they are
  % in fact,
  indeed essential for efficient storage emulations. In this paper, we
  answer this question in the affirmative. We show that relaxing the
  underlying storage to only support read/write operations leads to a
  linear blow-up in the emulation space requirements.
  % amount of space used by the emulation.
  % Specifically, we prove that no reilable emulation of a safe
  % multi-writer read/write register from a collection of fault-prone
  % servers supporting multi-writer/multi-reader atomic registres can
  % use fewer than $kf$ registers where $k$ is the number of supported
  % clients, and $f$ is the server failure threshold.
  We also show that the space complexity is not adaptive to
  concurrency, which implies that the storage cannot be reliably
  reclaimed even in sequential runs. On a positive side, we show that
  Compare-and-Swap primitives, which are commonly available with many
  off-the-shelf storage services, can be used to emulate a reliable
  multi-writer atomic register with constant storage and adaptive time
  complexity.

\end{abstract}

\vspace*{3cm}

\noindent\textbf{Corresponding author:}\\
\noindent Gregory Chockler\\
\noindent Royal Holloway, University of London\\
\noindent Egham TW11 0RP 
\noindent United Kingdom\\
\noindent \textbf{tel:} +44 (0)1784 443690\\
\noindent \texttt{Gregory.Chockler@rhul.ac.uk}\\\\

\noindent\textbf{Regular submission}\\

\noindent\textbf{Eligible to be considered for the best student paper
  award}: Alexander Spiegelman is a full-time student

\newpage
\setcounter{page}{1}

\section{Introduction}

Reliable storage emulations seek to construct fault-tolerant shared
primitives, such as read/write registers, from a collection of
failure-prone components, such as storage servers, or network-attached
disks. These emulations are core enablers of many modern storage
services and applications, such as cloud and online data
stores~\cite{pnuts,riak,spinnaker,zookeeper,mongodb} and
Storage-as-a-Service
offerings~\cite{s3,simple-db,dynamodb,azure-storage}.

Most existing emulation algorithms are constructed from storage
services capable of supporting custom-built read-modify-write (RMW)
primitives~\cite{ABD95,EnglertS00,EnglertS00,rambo,DuttaGLV10,GeorgiouNS09,DBLP:journals/eatcs/AguileraKMMS10}. For
example, the ABD algorithm~\cite{ABD95}, emulating a fault-tolerant
atomic read/write register from crash-prone nodes, assumes that each
node has an ability to test and update the stored data along with its
associated metadata in a single atomic step.  In reality though
reliable storage services must often be built from pre-existing or
off-the-shelf building blocks (such as network-attached disks or cloud
storage services), which typically offer a set collection of
read/write capabilities sometimes augmented with simple conditional
update primitives similar to Compare-and-Swap (CAS).

% lack the ability of supporting arbitrary RMW primitives.
% % can only be assumed to support a minimal set of read/write
% % interfaces potentially augmented with standard read-modify-write
% % primitives (such as Compare-and-Swap). 
% For example, the operations exposed by network-attached disks are
% typically limited to just basic read/write capabilities whereas the
% interfaces exposed by cloud storage services may further augment this
% with simple conditional update primitives similar to Compare-and-Swap
% (CAS).

In this paper, we study the question of what {\em minimal}\/
functionality must be supported by fault-prone storage nodes to enable
space-efficient emulations of reliable storage primitives. We start by
considering storage servers equipped with read/write primitives, which
we abstract as read/write atomic registers. A notable prior work
assuming a similar setting is Disk Paxos~\cite{2003:gafni}, which
builds a reliable consensus service from crash-prone network attached
disks. Interestingly, in Disk Paxos, each client is allocated a
dedicated register on each server, which naturally leads to the
question if linear space is necessary for constructing reliable
multi-writer storage from fault-prone read/write primitives.

% Interestingly, Disk Paxos~\cite{disk-paxos}, the only prior work we
% are aware of assuming a similar setting, builds a reliable consensus
% service from fault-prone network attached disks whereby each client is
% allocated a dedicated register on each server. It is therefore,
% natural to enquire whether the linear space complexity is necessary
% for constructing reliable multi-writer primitive from fault-prone
% read/write registers.

In Section~\ref{sec:space}, we prove that this is indeed inherent: the
number of registers required to implement a reliable multi-writer
read/write register for $k$ clients from a collection of multi-writer
multi-reader (MWMR) atomic read/write registers hosted on crash-prone
servers requires at least $kf$ registers where $f$ is the maximum
number of tolerated server failures. We further show that no such
algorithm can have its storage consumption adaptive to concurrency,
which implies that the storage costs cannot be further optimized
(e.g., by reclaiming old values) even in sequential runs. Since the
registers can be assigned to the servers in a variety of ways, we
further restrict possible assignments by showing that if the number of
registers per server is bounded by a known constant $m$, then
supporting $\ell m$ clients requires $f + 1$ more servers in addition
to the requisite $\ell f$ servers stipulated by our storage bound. Our
bounds apply to any fault-tolerant implementations of a MWMR register,
which are at least {\em single-writers safe}\/ (a consistency notion
weaker than the standard multi-writer safety~\cite{lamport1,mwr}), and
solo-terminating (a weak liveness condition where only the operations
eventually run in isolation are required to terminate).

We prove our results in a fault-prone shared memory
model~\cite{JCT98,faultyMemoryAfek1993benign,2006:abraham}, which
faithfully captures the settings where constituent storage services
are provided as pre-existing building blocks. Our impossibility proofs
employ a variation of a covering argument~\cite{burns-lynch} to
construct a sequential run where $f$ new registers become covered with
each consecutive write invoked by a client thus gradually exhausting
the available storage capacity.

Understanding the cost of using read/write primitives, we turn our
attention to identifying a simple RMW primitive that can be used to
efficiently support a reliable emulation. We focus on Compare-and-Swap
(CAS), which closely matches a variety of conditional write primitives
available with many of the today's cloud storage service
interfaces~\cite{pnuts,simple-db,dynamodb,azure-storage,mongodb}. In
Section~\ref{sec:rwalg}, we present a constant space emulation of a
MWMR atomic read/write register that utilizes a single CAS object per
server, and tolerates up to a minority of server crashes.  Our
emulation is derived in a modular fashion by first constructing the
ABD update primitive from a single CAS object, and then plugging the
resulting construction into the multi-writer ABD
emulation~\cite{ABD95,rambo}.  We show that the time complexity our
implementation matches that of ABD in contention-free runs, and, at
the worst case, is adaptive to the number of concurrently executing
clients.

% \input{related}

%\newpage
\section{Preliminaries}
\label{sec:prelim}

\subsection{Model}

We consider an {\em asynchronous fault-prone shared memory
  system}~\cite{JCT98} consisting of a set of {\em base}\/ objects
${\cal B} = \{b_1,b_2,\dots\}$. The objects are accessed by {\em
  clients}\/ from some set ${\cal C} = \{c_1, c_2, \dots\}$. The
clients interact with base objects via a set of operations supported
by the objects. We will consider base objects supporting either simple
{\em read}\/ and {\em write}\/ (i.e., read/write registers) or {\em
  compare-and-swap (CAS)}\/ operations.

We consider a slight generalization of the model in~\cite{JCT98} where
the objects are mapped to a set ${\cal S} = \{s_1, s_2, \dots\}$ of
servers via a function $\delta$ from ${\cal B}$ to ${\cal S}$. For
$B \subseteq {\cal B}$, we will write $\delta(B)$ to denote the {\em
  image}\/ of $B$, i.e., $\delta(B) = \{\delta(b) : b \in B\}$.
Conversely, for $S \subseteq {\cal S}$, we will write $\delta^{-1}(S)$
to denote the {\em pre-image}\/ of $S$, i.e.,
$\delta^{-1}(S) = \{b : \delta(b) \in S\}$.  Both servers and clients
can fail by crashing. A crash of a server causes all objects mapped to
that server to instantaneously crash\footnote{Note that the original
  faulty shared model of~\cite{JCT98} can be derived from our model by
  choosing $\delta$ to be an injective function.}.
% Any subset of size $f$ of servers can fail by crashing. A crash of a
% server results in all objects mapped to that server to simultaneously
% experience a {\em non-responsive crash (NR-crash}\/ failure in a
% single atomic step.  An object experiencing an NR-crash failure
% behaves correctly until it fails, and, once it fails, it never
% responds to any invocation.

We study algorithms that emulate shared read/write registers to a set
of clients. Clients interact with the emulated register via high-level
read and write operations. To distinguish the high-level emulated
reads and writes from low-level base object access, we refer to the
former as \Read\/ and \Write. We say that high-level operations are
{\em invoked}\/ and {\em return} whereas low-level operations are {\em
  triggered}\/ and {\em respond}. A high-level operation consists of a
series of trigger and respond actions on base objects, starting with
the operation's invocation and ending with its return. Since base
objects are crash-prone, clients must be able to continue executing
without awaiting responses to previously issued operations. Thus, the
trigger actions occur locally at clients without involving any actual
interaction with their target base objects.  Once triggered a
low-level operation can then {\em take effect} (or, be {\em applied}
to) the base object state followed by a response being returned to the
client.

An algorithm $A$ defines the behavior of clients as deterministic
state machines where state transitions are associated with actions,
such as trigger/response of low-level operations. A {\em
  configuration}\/ is a mapping to states from system components,
i.e., clients and base objects. An {\em initial configuration}\/ is
one where all components are in their initial states.

A {\em run}\/ of algorithm $A$ is a (finite or infinite) sequence of
alternating configurations and actions, beginning with some initial
configuration, such that configuration transitions occur according to
$A$. We use the notion of time $t$ during a run $r$ to refer to the
configuration reached after the $t$\textsuperscript{th} action in
$r$. A {\em run fragment}\/ is a contiguous sub-sequence of a run. A
run is {\em write-only}\/ if it has no invocations of the high-level
read operations.

We say that a base object, client, or server is {\em faulty}\/ in a
run $r$ if it fails at some time in $r$, and correct, otherwise. A run
is {\em fair}\/ if (1) for every low-level operation triggered by a
correct client on a correct base object, there is eventually a
matching response, and (2) every correct client gets infinitely many
opportunities to both trigger a low-level operation and execute the
return actions. We say that a low-lever operation on a base object is
{\em pending}\/ in run $r$ if it was triggered but has no matching
response in $r$.

We say that a high-level operation $op_i$ {\em precedes}\/ a high-level
operation $op_j$ in a run $r$, denoted $op_i \prec_r op_j$, if $op_i$
returns before $op_j$ is invoked in $r$. Operations $op_i$ and $op_j$
are concurrent in a run $r$, if neither one precedes the other. A run
with no concurrent operations is {\em sequential}.

\subsection{Storage Service Definitions}

We study storage services emulating a {\em multi-writer/multi-reader
  (MWMR) register}, which stores values from a domain $\mathbb{V}$,
and offers an interface for invoking read and write
operations. Initially, the register holds some distinguished initial
value $v_0 \in \mathbb{V}$.  The sequential specification of the
register is as follows: A read returns the latest written value, or
$v_0$ if none was written.

\noindent %
{\bf Liveness} We consider the following liveness conditions that must
be satisfied in fair runs of an emulation algorithm. A {\em
  wait-free}\/ object is one that guarantees that every high-level
operation invoked by a correct client eventually returns, regardless
of the actions of other clients. A {\em solo-terminating}\/ object
guarantees that every high-level operation that takes steps in
isolation eventually returns.

\noindent
{\bf Safety} Two runs are {\em equivalent}\/ if every client performs
the same sequence of high-level operations in both, where operations
that are pending in one can be either included (with some response) in
or excluded from the other. A {\em linearization}\/ of a run $r$ is an
equivalent sequential run that satisfies $r$'s operation
precedence relation and the object's sequential specification. 

We consider the following safety requirements for an emulation
algorithm. A run of the emulation algorithm satisfies {\em
  atomicity}\/ if it has a linearization. An emulated object is {\em
  atomic}\/ (or, {\em linearizable}) if all its runs satisfy
atomicity. For our storage lower bound, we will also consider the
following weak safety guarantee: A run $r$ of the MWMR emulation
algorithm is {\em single-writers}\/ if no two write operations overlap
in $r$: i.e., for any two distinct writes $w_i$ and $w_j$ in $r$
either $w_i \prec_r w_j$ or $w_j \prec_r w_i$. A run $r$ of the MWMR
register emulation algorithm satisfies {\em safety}~\cite{lamport1} if
for every read $rd$ that returns in $r$ and does not overlap any
writes, there exists a linearization $L_{rd}$ of the subsequence of
$r$ consisting of all write operations in $r$ and $rd$. An emulated
MWMR register is {\em single-writers safe (SW-safe)}\/ if all its
single-writers runs satisfy safety.

For our space lower bound, we will restrict our attention to {\em
  single-reader (SR)}\/ emulations where only a single designated
client is allowed to read the emulated register.

\noindent
{\bf Fault-Tolerance} The emulation algorithm is $f$-tolerant if it
remains correct (in the sense of its safety and liveness properties)
as long as at most $f$ servers crash for a fixed $f > 0$.

\noindent
{\bf Complexity measures} The {\em resource consumption} of an
emulation algorithm $A$ in a (finite) run $r$ is the number of base
objects used by $A$ in $r$. The {\em resource complexity}~\cite{JCT98}
of $A$ is the maximum resource consumption of $A$ in all its runs. To
measure running time, we assume that each operation triggered on a
base object takes at most one unit of time to complete, and the local
computation delays are negligibly small. The {\em (asynchronous) time
  complexity}\/ of $A$~\cite{attiya-book} is then the maximum time
required by any client to complete the high-level object invocation.

\noindent
{\bf Adaptivity to Contention} Given a run fragment $r$ of an
emulation algorithm, the {\em point
  contention}~\cite{point-cont-afek,point-cont} of $r$, $\PntCont(r)$,
is the maximum number of clients that have an incomplete high-level
invocation after some finite prefix of $r$. Similarly, we use
$\PntCont(op)$ to denote $\PntCont(r_{op})$, where $r_{op}$ is the run
fragment including all events between the $op$'s invocation and
response.

The resource complexity of $A$ is {\em adaptive to point contention}\/
if there exists a function $M$ such that after all finite runs $r$ of
$A$, the resource consumption of $A$ in $r$ is bounded by
$M(\PntCont(r))$. Likewise, the time complexity of $A$ is {\em
  adaptive to point contention}\/ if there exists a function $T$ such
that for each client $c_i$, and operation $op$, the time to complete
the invocation of $op$ by $c_i$ is bounded by $T(\PntCont(op))$.

\section{Resource Complexity of Emulating SW-Safe MWSR Register}
\label{sec:space}

In this section, we prove that any $f$-tolerant emulation of a
solo-terminating multi-writer/single-reader (MWSR) SW-safe register
for $k$ clients from of a collection of MWMR atomic registers stored
on crash-prone servers has resource complexity $kf$. As there are many
possible ways in which these $kf$ registers can be mapped to the given
set of servers, we further restrict possible mappings by showing that
if the number of registers assigned to each server is at most $m$,
then for any $\ell > 0$, the number of servers required to support
$\ell m$ clients is at least $\ell f + f + 1$. In other words,
supporting that many clients requires extra $f+1$ servers in addition
to $\ell f$ stipulated by our resource complexity bound. For
completeness, we will also show that $2f+1$ servers are necessary
regardless of the individual server capacities though this bound can
also be derived from well-known results (e.g.,
~\cite{Attiya:1990:RAE:79147.79158,Bracha:1985:ACB:4221.214134}). Our
last result shows that the emulation resource complexity cannot be
adaptive to point contention.

Our proof exploits the fact that the environment is allowed to prevent
a pending low-level write from taking effect on the base object states
for arbitrary long.  As a result, a client cannot reliably store a
value in a base register having a pending write (by a different
client) as this write may take effect at a later time thus erasing the
stored value. %
We will reuse the terminology of~\cite{burns-lynch}, and refer to a
pending write operation $W$ on some base register $b$ % , which
% was triggered but
% has not yet taken effect 
% on $b$'s state 
as a {\em covering write}, and to $b$ as being {\em covered}\/ by $W$.

\noindent
For any time $t$ (following the $t$\textsuperscript{th} action) in a
run $r$ of the emulation algorithm we define the following:

\begin{itemize}

\item $C(t)$: the set of clients that have completed a high-level
  write operation on the emulated register at time $\le t$.

\item $Cov(t)$: the set of the base registers that have a covering
  low-level write at time $t$.

% \item $V^+(t)$: the set of all values that were written by the clients
%   in $C^+(t)$ and the initial value $v_0$. 

% \item $S_k^-(t)$: the first $k$ servers having a register in
%   $Cov^-(t)$.

\end{itemize}

% In the following, we will restrict out attention to the single-writers
% runs of $A$, that is, $C^-(t) \le 1$ at all times $t$.

% \begin{defn}[Quescient server] A server $s$ is {\em quescient}\/ to a
%   client $c$ at time $t$, if between the last time a low-level
%   operation by $c$ took effect on the state of some base register on
%   $s$ and $t$, no other operation by $c$ % that have been triggered by
%   % $c$ on any register on $s$ before $t$
%   has taken effect on the state of any register on $s$.
% \end{defn}

\noindent
We first prove the following key lemma:

\begin{lemma}
  For all $F\subseteq {\cal S}$ such that $|F|=f$, there exists a
  write-only sequential run $r_i$ of an $f$-tolerant algorithm that
  emulates an SW-safe solo-terminating MWSR register consisting of
  $i\ge 0$ complete high-level writes of values $v_1,\dots,v_i$ by $i$
  distinct clients $c_1,\dots,c_i$, and $t_i$ steps
  % where the environment behaves like $Ad(F)$ for $|F| = f$,
  such that $|Cov(t_i)| \ge if$, % , (3)
  % $|\delta(\Cov^+(t') \setminus \Cov^+(t))| \ge f$
  and $\delta(Cov(t_i)) \cap F = \emptyset$.
  \label{lem:exhaustive-run}
\end{lemma}

We construct $r_i$ inductively as follows. First, it is easy to see
that a run $r_0$ consisting of $t_0=0$ steps satisfies the lemma. Next,
fix an arbitrary set of servers $F$ such that $|F|=f$, and assume that
$r_{i-1}$ exists for all $i > 0$. We show how $r_{i-1}$ can be
extended up to time $t_i > t_{i-1}$ so that the lemma holds for the
resulting run.
% with another complete high-level write so that the lemma holds for
% $r_i$.

\noindent
We introduce the following notation for all times $t\ge t_{i-1}$:

\begin{itemize}

\item $Tr_i(t)$: the set of base registers which had a low-level write
  triggered on between $t_{i-1}$ and $t$.

\item $Cov_i(t) = Cov(t) \setminus Cov(t_{i-1})$: the set of base
  registers that have been newly covered between $t_{i-1}$ and
  $t$. Note that $Cov_i(t) \subseteq Tr_i(t)$.

\item $Q_i(t) \subseteq {\cal S}$: the set of servers such that
  $Q_i(t) = \delta(Cov_i(t)) \setminus F$ if
  $|\delta(Cov_i(t)) \setminus F| \le f$, and $Q_i(t) = Q_i(t-1)$,
  otherwise.

\end{itemize}

We will define the following adversarial behaviour of the environment,
which whilst being tolerated by the algorithm causes it to consume a
gradually growing amount of the storage resources:

\begin{defn}[$Ad_i$]: At any time $t\ge t_{i-1}$: prevent the
  following writes from taking effect on the base register states:
  \begin{enumerate}
  \item all covering writes by clients in $C(t_{i-1})$,
    % on the base
    % registers in
    % % $Cov(t_{i-1}) \cup \{b \in {\cal B} : \delta(b)\in Q_i(t)\}$
    % $Cov(t_{i-1})$,
    and
  \item all covering writes on the base registers in
    $\delta^{-1}(Q_i(t))$.
  \end{enumerate}
\end{defn}

\begin{observation}
  If the environment behaves like $Ad_i$, then for all
  $t \ge t_{i-1}$, $Q_i(t) \subseteq Q_i(t+1)$.
\label{obs:ad}
\end{observation}

We first show that $r_{i-1}$ can be extended with a complete
high-level write $W_i$ by a new client $c_i$ such that the environment
behaves like $Ad_i$ until $W_i$ returns. Intuitively, this means that
$Ad_i$ delays applying low-level writes triggered by $c_i$ on at most
$f$ servers as well as the past covering writes. As a result $c_i$
cannot distinguish this scenario from the one where all the involved
servers and clients have crashed, and therefore, by solo-termination,
must return without before receiving the delayed replies.

\begin{lemma}
  Suppose that the environment behaves like $Ad_i$, and let $W_i$ be a
  high-level write invocation by client $c_i \not\in C(t_{i-1})$. Then,
  there exists time $t_r > t_{i-1}$ at which $W_i$ returns while the
  environment continues to behave like $Ad_i$ until $t_r$.
  % It is possible to extend $r_{i-1}$ with a complete invocation of a
  % high-level write operation $W_i$ by client $c_i \neq c_{i-1}$ such
  % that the environment behaves like $Ad_i$ until $W_i$ returns.
  \label{lem:1write}
\end{lemma}

\begin{proof}
  By definition of $Ad_i$, there exists time $t_f > t_{i-1}$ such that
  for all times $t \ge t_f$, $Q_i(t) = Q_i(t_f)$. If $W_i$ returns
  before $t_f$, then $t_r = t_f$ satisfies the lemma. Otherwise, for
  each server $s\in Q_i(t_f)$, let $t_s$ be the earliest time such
  that $s\in Q_i(t_s)$. Since by Observation~\ref{obs:ad},
  $Q_i(t) \subseteq Q_i(t_f)$ for all $t \le t_f$, $s \in Q_i(t)$, for
  all $t \ge t_s$.

  Let $r'$ be a fair run, which includes the same sequence of steps as
  $r_{i-1}$ up to time $t_f$, and in addition, each server
  $s\in Q_i(t_f)$ fails immediately after the step $t_s$, and each
  client $c_1,\dots,c_{i-1}$ fails before any of its covering writes
  on registers in $Cov(t_{i-1})$ takes effect on the register
  states. Since $r'$ is fair, by $f$-tolerance and solo-termination,
  there exists time $t'$ at which $W_i$ returns in $r'$. Since
  $r_{i-1}$ is indistinguishable from $r'$ to $c_i$ for the entire
  duration of $W_i$, it must return in $r_{i-1}$ at time $t_r=t'$ as
  well.
\end{proof}

We next show that in order to guarantee correctness in the face of the
environment behaving like $Ad_i$, $W_i$ must trigger a low-level write
on at least one non-covered base register on each server in a set of
$2f+1$ servers.
% The intuition is that by $Ad_i$ definition, $f$ servers are
% non-responsive to $c_i$, and another $f$ servers may crash after $W_i$
% completes, so $W_i$ must trigger a low-level write on some registers
% on at least $2f + 1$ servers in order that later read will be able to
% return its value.

\begin{lemma}
  Let $W_i$ be a high-level write invocation by client
  $c_i \not\in C(t_{i-1})$ that returns at time $t_r > t_{i-1}$, and
  suppose that the environment behaves like $Ad_i$ until $t_r$. Then,
  $|\delta(Tr_i(t_r) \setminus Cov(t_{i-1}))| > 2f$.
\label{lem:2f}
\end{lemma}

\begin{proof}
  % Assume by contradiction that
  % $|\delta(Tr_i(t_r) \setminus Cov(t_{i-1}))| = 2f$. 
  % % $Q_i(t_r) \ge f$
  % For each server $s\in \delta(Tr_i(t_r) \setminus Cov(t_{i-1}))$, let
  % $t_s$ be the earliest time a low-level write is triggered on a
  % register $b \in \delta^{-1}(\{s\})$, and $b\not\in Cov(t_{i-1})$.
  % Since $|F|=f$, we can partition
  % $\delta(Tr_i(t_r) \setminus Cov(t_{i-1}))$ into two sets $S_1$ and
  % $S_2$ such that $|S_1|=|S_2|=f$, $S_1$ includes all servers in
  % $F \cap \delta(Tr_i(t_r) \setminus Cov(t_{i-1}))$, and the remaining
  % servers $s \in S_1$ satisfy $t_s > t_{s'}$ for all servers
  % $s' \in S_2$.
  Denote $M \triangleq \delta(Tr_i(t_r) \setminus Cov(t_{i-1}))$, and
  assume by contradiction that $|M| \leq 2f$.  Let $S_1=M \cap F$,
  $S_2=Q_i(t_r)$, and $S_3=M \setminus (S_1\cup S_2)$.  Note that
  $S_1,S_2,S_3$ are pairwise disjoint, $M = S_1 \cup S_2 \cup S_3$,
  and by definition of $Q_i(t_r)$, and since $|F|=f$,
  $|S_1 \cup S_3| = |S_1|+|S_3| \leq f$.

  Let $r$ be a run, which is identical to $r_{i-1}$ up to time
  $t_{i-1}$, after which all the covering writes in $r_{i-1}$ take
  effect on register states, and all servers in the set $S_1 \cup S_3$
  crash. Extend $r$ with an invocation of a high-level read operation
  $R$ by client $c_{rd}\neq c_i$. Since $r$ is fair, by
  solo-termination and $f$-tolerance, there exists time
  $t_{rd} > t_{i-1}$ at which $R$ returns. Since $r$ is
  single-writers, by SW-safety, $R$ must return $v_{i-1}$.

  Let $r'$ be a run, which is identical to $r_{i-1}$ up to time $t_r$,
  after which it is extended to time $t' > t_r$ by having all servers
  in the set $S_1 \cup S_3$ crash, and the covering writes in
  $r_{i-1}$ to take effect on the base register states. As a result,
  the values stored in the registers in $Cov(t_{i-1})$ are now
  identical to those in $r$.  Furthermore, since $Ad_i$ prevents all
  low-level writes triggered on registers in $\delta^{-1}(S_2)$ from
  taking effect before $t_r$, their values are also the same as those
  in $r$.  Thus, at $t'$, all registers in both $r$ and $r'$ have the
  same content.

  % Let $r'$ be a run, which is identical to $r_{i-1}$ up to time $t_r$,
  % after which it is extended to time $t' > t_r$ by having all the
  % covering writes in $r_{i-1}$ to take effect on the base register
  % states. As a result, the values stored in the registers in
  % $Cov(t_{i-1})$ are now identical to those in $r$. Furthermore, since
  % $Ad_i$ prevents all low-level writes triggered on registers in
  % $\delta^{-1}(S_2)$ from taking effect, their values are also the
  % same as those in $r$.  Thus, at $t'$, the only registers whose
  % respective values may differ from those in $r$ are the registers in
  % $\delta^{-1}(S_1)$.

  We extend $r'$ by having client $c_{rd} \neq c_i$ to invoke
  high-level read $R$ while allowing the environment to continue
  preventing all covering writes by client $c_i$ on the registers in
  $\delta^{-1}(S_2)$ from taking effect on their states. Since $r'$ is
  indistinguishable from $r$ to $c_{rd}$, the sequence of steps
  executed by $c_{rd}$ in $r'$ is the same as that in $r$.  Hence, $R$
  returns $v_{i-1}$ in $r'$. However, since $W_i$ is the last complete
  write preceding $R$ in $r'$, by SW-safety, the $R$'s return value
  must be $v_i \neq v_{i-1}$. A contradiction.

  % We extend $r'$ by having client $c_{rd} \neq c_i$ to invoke
  % high-level read $R$ while allowing the environment to both delay
  % responses to all low-level operations triggered by $c_{rd}$ on the
  % registers in $\delta^{-1}(S_1)$, and continue preventing all
  % covering writes on the registers in $\delta^{-1}(S_2)$ from taking
  % effect on their states. Since $r'$ is indistinguishable from $r$ to
  % $c_{rd}$, the sequence of steps executed by $c_{rd}$ in $r'$ is the
  % same as that in $r$. Hence, $R$ returns $v_{i-1}$ in $r'$. However,
  % since $W_i$ is the last complete write preceding $R$ in $r'$, by
  % SW-safety, the $R$'s return value must be $v_i \neq v_{i-1}$. A
  % contradiction.
\end{proof}

\noindent
The following two corollaries follow immediately from
Lemmas~\ref{lem:1write} and \ref{lem:2f}:

\begin{cor}
  Let $W_i$ be a high-level write invocation by client
  $c_i \not\in C(t_{i-1})$ that returns at time $t_r > t_{i-1}$, and
  suppose that the environment behaves like $Ad_i$ until $t_r$. Then,
  $Q_i(t_r) = f$.
\label{cor:Q}
\end{cor}

\begin{cor}
  For all $i > 0$, $|{\cal S} \setminus \delta(Cov(t_{i-1}))| > 2f$.
\label{cor:2f}
\end{cor}

\noindent
We are now ready to prove Lemma~\ref{lem:exhaustive-run}:

\begin{proof}[of Lemma~\ref{lem:exhaustive-run}]
  By Lemma~\ref{lem:1write}, $r_{i-1}$ can be extended with a complete
  high-level write $W_i$ by client $c_i \neq c_{i-1}$ writing a value
  $v_i \neq v_{i-1}$ while allowing the environment to behave like
  $Ad_i$ until time $t_r$ when $W_i$ returns. We further extend
  $r_{i-1}$ by allowing the environment to behave like $Ad_i$ until
  time $t' > t_r$ when all writes triggered after $t_{i-1}$ on the
  registers in $\delta^{-1}(F)$ take effect. Hence,
  $F \cap \delta(Cov_i(t')) = \emptyset$.

  % By Lemma~\ref{lem:2f},
  % $|\delta(Tr_i(t_r) \setminus Cov(t_{i-1}))| > 2f$. By definition of
  % $Ad_i$, there exist $f$ servers $s$ in
  % $\delta(Tr_i(t_r) \setminus Cov(t_{i-1}))$ such that all registers
  % that have been newly covered on $s$ after $t_{i-1}$ remain covered
  % at $t'$. Thus, $Cov_i(t') \ge f$.

  Since by Corollary \ref{cor:Q}, $Q_i(t_r)=f$, and by
  Observation~\ref{obs:ad}, $Q_i(t_r) \subseteq Q_i(t')$, $Q_i(t')=f$,
  and therefore, $|Cov_i(t')| \geq f$. Now since $Cov_i(t')$ and
  $Cov(t_{i-1})$ are disjoint,
  $Cov(t') = Cov(t_{i-1}) \cup Cov_i(t')$, and by the induction
  hypothesis $|Cov(t_{i-1})| \ge (i-1)f$, and
  $\delta(Cov(t_{i-1})) \cap F = \emptyset$, we receive
  $|Cov(t')| \ge (i-1)f + f = if$, and
  $\delta(Cov(t')) \cap F = (\delta(Cov(t_{i-1})) \cap F) \cup
  (\delta(Cov_i(t')) \cap F) = \emptyset$.
  Thus, $t_i = t'$ satisfies the lemma.
\end{proof}

% We now proceed to proving our impossibility results.
\noindent
\textbf{Resource Complexity} The following theorem follows immediately
from Lemma~\ref{lem:exhaustive-run} (please see
Section~\ref{sec:space:app} of the Appendix for a full proof):
% We first show that any $f$-tolerant algorithm emulating an SW-safe
% solo-terminating MWSR register for $k\ge 0$ clients must use at
% least $kf$ base registers.

\begin{thm}
  For any $k\ge 0$, $f\ge 0$, there is no $f$-tolerant algorithm
  emulating an SW-safe solo-terminating MWSR register for $k$ clients
  using less than $kf$ base registers.
\label{thm:kf}
\end{thm}

% \begin{proof}
%   Pick arbitrary $k\ge 0$, $f\ge 0$, and assume by contradiction that
%   there exists an $f$-tolerant algorithm $A$ that emulates an SW-safe
%   solo terminating MWSR register for $k$ clients with fewer than $kf$
%   base registers. By Lemma~\ref{lem:exhaustive-run}, there exists a
%   run $r$ of $A$ consisting of $k$ high-level writes by $k$ distinct
%   clients such that by the end of $r$, the number of distinct base
%   registers having a covering write is at least $kf$. Hence, $A$ will
%   require at least $kf$ distinct base registers to support $k$
%   clients. A contradiction. 
% \end{proof}

\noindent
\textbf{Number of Servers} We now turn our attention to deriving the
number of servers required for supporting the emulation. The following
result follows immediately from Corollary~\ref{cor:2f} (please see
Section~\ref{sec:space:app} of the Appendix for a full proof), but can
also be derived from well-known results in the literature (e.g.,
~\cite{Attiya:1990:RAE:79147.79158,Bracha:1985:ACB:4221.214134})

% We first show that any $f$-tolerant MWSR emulation of an SW-safe
% register for any number of clients requires at least $2f+1$ servers
% even if each server can store an unbounded number of registers.

\begin{thm}
  For any $k > 0$, and $f \ge 0$, there is no $f$-tolerant algorithm
  emulating an SW-safe solo-terminating MWSR register for $k$ clients
  with less than $2f+1$ servers.
\label{thm:2f+1}
\end{thm}

% \begin{proof}
%   Assume by contradiction that there exists an $f$-tolerant algorithm
%   emulating an SW-safe solo-terminating MWSR register for $k>0$
%   clients using $2f$ servers. By Corollary~\ref{cor:2f}, there exists
%   a run $r_1$ of $A$ consisting of a single high-level write $W_1$ by
%   a client $c_1$ such that
%   $|{\cal S} \setminus \delta(Cov(t_0))| > 2f$ where $t_0=0$. Since no
%   base registers are covered at $t_0$,
%   $|{\cal S} \setminus \delta(Cov(t_0))| = |{\cal S}| > 2f$. However,
%   by assumption, $|{\cal S}| = 2f$. A contradiction.
% \end{proof}

Next, we show that if the storage per server is bounded by a known
constant, an extra $f+1$ servers beyond the minimum capacity
established by Theorem~\ref{thm:kf} are necessary to accommodate a
given number of clients.

\begin{thm}
  For any $m > 0$, $\ell > 0$, and $f\ge 0$, there is no $f$-tolerant
  algorithm emulating an SW-safe solo-terminating MWSR register for
  $k \ge \ell m$ clients using less than $\ell f + f + 1$ servers if
  each server can store at most $m$ registers.
\label{thm:servers}
\end{thm}

\begin{proof}
  Assume by contradiction there exists an $f$-tolerant algorithm $A$
  emulating an SW-safe solo-terminating MWSR register for $k=\ell m$
  clients using $\ell f + f$ servers.  Fix a set $F \subseteq S$, such
  that $|F|=f$, and let $N \leq mf$ be the number of registers mapped
  to the servers in $F$.

  By Lemma~\ref{lem:exhaustive-run}, there exists a run $r_{k-1}$ of
  $A$ consisting of $k-1 = \ell m - 1$ high-level writes by $k-1$
  distinct clients such that by the end of $r_{k-1}$, the number of
  distinct base registers having a covering write is at least
  $(k-1)f$, and no registers in $\delta^{-1}(F)$ have a covering
  write.  Thus, the number of registers that remain not covered by the
  end of $r_{k-1}$ is at most
  $\ell fm + N - (k-1)f = \ell f m + N - \ell f m + f = N + f \triangleq
  R$.

  Now since no register in $\delta^{-1}(F)$ has a covering write, $N$
  out of total $R$ registers must be mapped to the $f$ servers in $F$.
  And since the remaining $f$ registers can be mapped to at most $f$
  servers, by the end of $r_{k-1}$, the total number of servers that
  may have a register without a covering write is at most $2f$. A
  contradiction to Corollary~\ref{cor:2f}.
\end{proof}

\noindent
\textbf{Adaptivity} We show that no SW-safe solo-terminating MWSR
register can have a fault-tolerant emulation adaptive to point
contention:

\begin{thm}
  For any $f>0$, there is no $f$-tolerant algorithm that emulates an
  SW-safe solo-terminating MWSR register with resource complexity
  adaptive to point contention.
\label{thm:adaptive}
\end{thm}

\begin{proof}
  Pick an arbitrary $f>0$, and assume by contradiction that such an
  algorithm $A$ exists. By Lemma~\ref{lem:exhaustive-run}, there
  exists a run $r$ of $A$ consisting of $k$ high-level writes by $k$
  distinct clients such that the resource complexity grows by $f$ for
  each consecutive write that completes in $r$ whereas the point
  contention remains equal $1$ for the entire $r$. We conclude that no
  function mapping point contention to resource consumption can exist,
  and therefore, $A$'s resource complexity is not adaptive to point
  contention. A contradiction.
\end{proof}
%\newpage

%\input{linear-time-lb}

\section{Atomic Register Implementation}
\label{sec:rwalg}

In this section we present a space-efficient $f$-tolerant algorithm
implementing a wait-free MWMR atomic register from a collection of
$n>2f$ servers each storing a single \CAS\/ object. Unlike previous
space-efficient approaches our algorithm does not require support for
any specialized read-modify-write functionality besides $\CAS$, i.e.,
conditional write, obviating the need for a custom server code.
% The algorithm
% uses a single single \CASobj\ object per server, and its time
% complexity is adaptive to concurrency. 
The algorithm's time complexity is {\em adaptive}\/ to concurrency
guaranteeing that each operation $op$ terminates in at most $O(c^2)$
steps where $c=\PntCont(op)$.
% where $c$ is the maximum number of clients executing concurrently with
% $op$ (i.e., $c=\PntCont(op)$).

Our algorithm, called {\em CAS-ABD}, is derived from the multi-writer
ABD~\cite{ABD95} emulation of an atomic read/write register to which
we refer as {\em MW-ABD}. For completeness, the MW-ABD implementation
is briefly reviewed in Section~\ref{sec:mw-abd} below (full details
can be found in~\cite{rambo}). The CAS-ABD algorithm is described in
Section~\ref{sec:cas-abd}.

\subsection{MW-ABD Algorithm}
\label{sec:mw-abd}

The MW-ABD shared state consists of a set ${\cal B}$ of $n > 2f$
crash-prone objects $\{b_1,\dots,b_n\}$ mapped to a set ${\cal S}$ of
$n$ servers ${\cal S} = \{s_1,\dots,s_n\}$ such that
$\delta(b_i) = s_i$ for each $1 \le i \le n$. Each object $b_i$ stores
a pair $(ts, val)$ where $ts$ is a timestamp and $val \in \mathbb{V}$.
We will write $b_i.ts$ and $b_i.val$ to refer to the timestamp and
value components of $b_i$ respectively. Each timestamp $ts$ is a pair
$(num, c)$ where $num\in \mathbb{N}$ is a natural number, and
$c\in {\cal C}$ is a client. We will write $ts.num$ and $ts.c$ to
refer to the $ts$'s first and second component respectively. The
timestamps are ordered lexicographically so that $ts < ts'$ if
$ts.num < ts'.num$, or $ts.num = ts'.num$ and $ts.c < ts.c'$.  The
MW-ABD types and shared states are summarized in
Algorithm~\ref{alg:ABD-states}.

\begin{algorithm}[H]
 \caption{Types and States of MW-ABD and CAS-ABD}
 \label{alg:ABD-states}
\begin{algorithmic}[1]
  \footnotesize %
  \State $TS = \mathbb{N} \times {\cal C}$, the set of timestapms with
  selectors $num$ and $c$%
  \State $TSVal = TS \times \mathbb{V}$, with selectors $ts$ and
  $val$ %
  \State ${\cal B} = \{b_1,\dots,b_n\}$: the set of shared objects
  such that $b_i \in TSVal$ for all $1 \le i \le n$; initially
  $b_i = ((0, 0), v_0)$
\end{algorithmic}
\end{algorithm}

The sequential specification supported by each object
$b_i\in {\cal B}$ is shown in Algorithm~\ref{alg:updateABD}.  It
consists of two atomic operations: $read$ and $update$. The read
operation returns the current content of $b_i$ (i.e.,
$(b_i.ts, b_i.val)$); and the $update$ operation is a
read-modify-write (RMW) primitive comprised of atomically executed
sequence of steps shown in
lines~\ref{line:update-start}--\ref{line:update-end} of
Algorithm~\ref{alg:updateABD}. We henceforth refer to the object type
supporting the sequential specification in
Algorithm~\ref{alg:updateABD} as {\em ABD Object (ABDO)}.

\begin{algorithm}[H]
  \caption{The ABDO sequential specification for each $b_i$,
    $1 \le i\le n$}
 \label{alg:updateABD}
\footnotesize
\begin{algorithmic}[1]
\begin{multicols}{2}
\Operation{$update(b_i, t, v)$}{} 
\State {\bf if} $b_i.ts < t$ \label{line:update-start}
\State \hspace{0.4cm} $b_i.ts \gets t$
\State \hspace{0.4cm} $b_i.val \gets v$
\State {\bf return} ack \label{line:update-end}
\EndOperation

\columnbreak

\Operation{$read(b_i)$}{}
\State \textbf{return} $(b_i.ts, b_i.val)$
\EndOperation
\end{multicols}
\end{algorithmic}
% \vspace*{-0.4cm}
\end{algorithm}

%\vspace*{-0.5cm}
The implementation of both write and read proceeds by invoking
consecutive {\em rounds}\/ of base object accesses. At each round, the
client triggers operations on all base objects in parallel, and awaits
responses from at least $n-f$ objects. The write implementation
consists of two rounds. In the first round, the writer collects the
set $R$ of $(b_i.ts,b_i.val)$ pairs from $n-f$ objects by triggering
$b_i.read$ on all objects $b_i\in {\cal B}$. The writer then
determines a new timestamp $ts'$ to be stored alongside the value $v$
being written so that
$ts'.num = \max \{num' : (num',\ast) \in R \} + 1$ and $ts'.c$ is the
writer's identifier. This is followed by another round where the
writer triggers $b_i.update(b_i, ts, v)$ on each base object $b_i$ to
replace its current content with $(ts, v)$.

The first round of read is identical to that of write except that the
set $R$ is used to identify the value $v' \in \mathbb{V}$ having the
highest timestamp $ts'$ among the timestamp/value pairs in $R$. This
is followed by another round where the reader invokes
$b_i.update(b_i, ts', v')$ on each base object $b_i$ to ensure
$(ts',v')$ is available from all sets of $n-f$ base objects.  The
reader then returns $v'$.
% This ensures that the order of the timestamps associated with the
% stored values is consistent with the order of their corresponding
% writes (which is essential for atomicity)

\subsection{CAS-ABD Algorithm}
\label{sec:cas-abd}

Suppose that the base ABD objects in ${\cal B}$ are substituted with
{\em Compare-and-Swap (CAS)}\/ objects: i.e., the sequential
specification of each $b_i \in {\cal B}$ consists of a single \CAS\/
primitive whose code is shown in
lines~\ref{line:cas-start}--\ref{line:cas-end} of
Algorithm~\ref{alg:cas-abd}. We obtain an implementation of an
$f$-tolerant MWMR atomic read/write register from a collection of
$n>2f$ \CAS\/ base objects, to which we refer as {\em CAS-ABD}, in a
modular fashion by first constructing an ABDO from a {\em single}\/
\CAS\/ base object $b_i$ using the emulation algorithm in
Algorithm~\ref{alg:cas-abd}, and then, plugging the resulting
construction into the MW-ABD algorithm described above.

\begin{algorithm}[H]
  \caption{The ABDO emulation from a single \CAS\/ object $b_i$,
    $1\le i\le n$}
 \label{alg:cas-abd}
\footnotesize

\begin{algorithmic}[1]
\begin{multicols}{2}
\Statex Local variables: 
\Statex \hspace*{0.4cm} $exp \in TSVal$, initially $((0, 0), v_0)$

\Operation{$update(b_i, t, v)$}{}
\State $done \gets false$
\State \textbf{if} $t > exp.ts$  \label{line:test}
\State \hspace*{0.4cm} {\textbf{repeat}} \label{line:loop-start}%
\State \hspace*{0.4cm} \hspace*{0.4cm} $old \gets \CAS(b_i, exp, (t, v))$ \label{line:cas} %
\State \hspace*{0.4cm} \hspace*{0.4cm} \textbf{if} $old = exp \vee old.ts \ge t$ \label{line:aftercas}%
\State \hspace*{0.4cm} \hspace*{0.4cm} \hspace*{0.4cm} $done \gets true$ %
\State \hspace*{0.4cm} \hspace*{0.4cm} $exp \gets old$ \label{line:exp-old} %
\State \hspace*{0.4cm} \textbf{until} $done \gets true$ \label{line:loop-end} %
\State \textbf{return} ack
\EndOperation

\columnbreak

\Operation{$read(b_i)$}{}
\State \textbf{return} $\CAS(b_i, exp, exp)$ \label{line:read-cas}
\EndOperation

\vfill

\Operation{$\CAS(b_i, exp, new)$, $exp, new \in TSVal$}{} \label{line:cas-start}
\State $prev \gets b_i$
\State \textbf{if} $exp = b_i$
\State \hspace*{0.4cm} $b_i \gets new$
\State {\bf return} $prev$ \label{line:cas-end}
\EndOperation
\end{multicols}
\end{algorithmic}
%\vspace*{-0.4cm}
\end{algorithm}

%\vspace*{-0.4cm}
In order to prove that CAS-ABD is a correct implementation of an
$f$-tolerant wait-free MWMR read/write register, it suffices to show
that the ABDO emulation in Algorithm~\ref{alg:cas-abd} is a wait-free
linearizable implementation of the ABD object. Below we show that this
is indeed the case assuming that the following property, to which we
henceforth refer as {\em timestamp uniqueness}, is satisfied in all
runs $r$ of ABDO: for all objects $b_i \in {\cal B}$, $r$ includes at
{\em most}\/ one invocation of the form $update(b_i, ts, \ast)$. Given
that linearizability is a composable property~\cite{herlihy}, and
MW-ABD is known to satisfy timestamp-uniqueness in all runs, the
correctness of CAS-ABD then follows from the correctness of
MW-ABD~\cite{rambo}.

To show linearizability~\cite{herlihy}, we first identify for each
invocation of $update$ and $read$ in each possible run of the ABDO
emulation, a single step within the operation execution, called a {\em
  linearization point}\/ (i.e., a single step where the operation
takes effect on the base object state), as follows: For each $read$
invocation, the linearization point is simply the return step in
line~\ref{line:read-cas}. The linearization points for the $update$
invocations are assigned to either one of the following two steps: (1)
if $update$ returns without entering the loop in
lines~\ref{line:loop-start}--\ref{line:loop-end}, the condition test
step in line~\ref{line:test} is the linearization point; and (2) if
$update$ returns due to the condition in line~\ref{line:aftercas}
being true, then the \CAS\/ call in line~\ref{line:cas} is the
linearization point. The linearizability then follows from following
lemma (proven in Section~\ref{sec:rwcorrectness} of the Appendix),
which asserts that the sequence obtained by shrinking each operation
to occur atomically at its linearization point is a valid sequential
run of ABDO.

\begin{lemma}
  Let $r$ be a run of the ABDO emulation in
  Algorithm~\ref{alg:cas-abd}, and $\sigma$ be a sequential run
  obtained from $r$ by shrinking each $update$ and $read$ operation to
  occur at its linearization point. Then, $\sigma$ is a sequential run
  of the ABD object in Algorithm~\ref{alg:updateABD}.
\label{lem:linear}
\end{lemma}

Since the $read$ implementation is obviously wait-free, we only need
to argue wait freedom for the $update$ operations. To see this,
observe that $t > exp.ts$ 
%and $b_i.ts \ge exp.ts$ 
every time before \CAS\/ is called in line~\ref{line:cas} (see
Lemmas~\ref{lem:exp} in Section~\ref{sec:rwcorrectness} of the
Appendix). Since $b_i.ts = exp.ts$ is a necessary condition for a
successful \CAS\/ call, the value of $b_i$ can only be changed when
$t > b_i.ts$. Hence, the timestamps of the values stored in each $b_i$
are non-decreasing (see Lemma~\ref{lem:non-decreasing} in
Section~\ref{sec:rwcorrectness} of the Appendix). If $b_i.ts$ does not
change between the consecutive iterations of the loop in
lines~\ref{line:loop-start}--\ref{line:loop-end}, timestamp uniqueness
implies that the next call to \CAS\/ will be successful and the loop
terminates. Otherwise, the fact that the timestamps are non-decreasing
implies that $b_i.ts$ is superseded by a higher timestamp. Since there
are only finitely many timestamps lower than $t$, the loop will
terminate no later than the value of $b_i.ts$ reaches or exceeds
$t$. Thus, we have the following result (see
Section~\ref{sec:rwcorrectness} of the Appendix for the full proof):

\begin{lemma}
  The ABDO emulation in Algorithm~\ref{alg:cas-abd} is wait-free
  provided all its runs satisfy timestamp uniqueness.
\label{lem:wait-free}
\end{lemma}

\noindent
Given that timestamp uniqueness holds in all runs of MW-ABD, we
receive the following:

\begin{thm}
  The CAS-ABD algorithm is an $f$-tolerant implementation of a
  wait-free MWMR atomic register.
\label{thm:cas-abd}
\end{thm}

{\bf Time Complexity} It is easy to see that in the absence of
contention, the $update$ operation terminates in at most $2$ rounds of
the base object accesses. This can be further optimized if the clients
keep a local copy of the most recent value read from each object $b_i$
at the read round of CAS-ABD, and then use this value to initialize
the expected value parameter $exp$ of \CAS.  Thus, in the best case
scenarios when the object replies are received in a timely fashion,
and there is no contention, $update$ will terminate in just $1$ round,
thus achieving the $2$ round complexity of MW-ABD overall.

In the presence of contention, the number of unsuccessful \CAS\/ calls
executed within the update operation loop in
lines~\ref{line:loop-start}--\ref{line:loop-end} is bounded by the
number of unique timestamps returned by the \CAS\/ calls that are
smaller than the timestamp $t$ supplied to the update. Given the way
the timestamps are chosen by the algorithm, the number of such
timestamps per each of the $c$ concurrently executing clients is
constant. However, since the $num$ component of each timestamp can be
shared by concurrently executing clients, the overall time complexity
of update can be as high as $c^2$. In Section~\ref{sec:rwcorrectness}
of the Appendix, we prove that $c$ is equal to the maximum number of
clients that can execute concurrently with the update thus obtaining
the following:

\begin{thm}
  The CAS-ABD time complexity is {\em adaptive}\/ to concurrency
  guaranteeing that each operation $op$ terminates in at most $O(c^2)$
  base object accesses where $c=\PntCont(op)$.
\label{thm:time}
\end{thm}

\section{Conclusions and Future Work}

We studied the resource complexity of emulating an $f$-tolerant
read/write MWMR register from a collection of atomic MWMR registers
stored on crash-prone servers. We established a number of lower
bounds that apply to any fault-tolerant emulation of a MWMR register,
which satisfies weak correctness guarantees: single-writers safety,
and solo-termination. In particular, we proved that no such emulation
can use fewer than $kf$ registers to support $k>0$ clients or have its
storage consumption adaptive to concurrency. We also characterized
possible allocations of registers to servers by showing that if the
number of registers per server is bounded by a known constant $m$,
then supporting $\ell m$ clients requires $f + 1$ more servers in
addition to the requisite $\ell f$ servers implied by our storage
bound.

In search for a simple RMW primitive that can be leveraged for
obtaining a space-efficient implementation, we studied reliable
storage emulations from crash-prone CAS objects. To this end, we
presented a constant space emulation of an MWMR atomic read/write
register that utilizes a single CAS object per server, tolerates up to
a minority of server crashes, and has time complexity adaptive to
point contention.
 
Our work leaves some questions open for future work. First, observe
that ABD can be applied in a straightforward fashion to implement an
MWMR wait-free atomic register from fault-prone registers by assigning
each client to a dedicated set of $2f+1$ registers stored on $2f+1$
different servers. An interesting open question is then whether our
lower bound can be further tightened to match this storage cost, or
there are emulations that can achieve a tighter storage cost (e.g., by
weakening their correctness guarantees). Second, the worst-case time
complexity of our CAS-based ABD implementation is quadratic in point
contention. It will be interesting to explore whether it can be
further improved (e.g., by modifying the ABD timestamp selection
mechanism), or this is an inherent limitation.

\newpage
\bibliographystyle{IEEEtran}
\bibliography{lit}

\newpage
\appendix
%%%%%%%%%%%%%%%%%
% \newpage
\section{Space Lower Bounds}
\label{sec:space:app}

\begin{thm}
  For any $k\ge 0$, $f\ge 0$, there is no $f$-tolerant algorithm
  emulating an SW-safe solo-terminating MWSR register for $k$ clients
  using less than $kf$ base registers.
\label{thm:app:kf}
\end{thm}

\begin{proof}
  Pick arbitrary $k\ge 0$, $f\ge 0$, and assume by contradiction that
  there exists an $f$-tolerant algorithm $A$ that emulates an SW-safe
  solo terminating MWSR register for $k$ clients with fewer than $kf$
  base registers. By Lemma~\ref{lem:exhaustive-run}, there exists a
  run $r$ of $A$ consisting of $k$ high-level writes by $k$ distinct
  clients such that by the end of $r$, the number of distinct base
  registers having a covering write is at least $kf$. Hence, $A$ will
  require at least $kf$ distinct base registers to support $k$
  clients. A contradiction. 
\end{proof}

\begin{thm}
  For any $k > 0$, and $f \ge 0$, there is no $f$-tolerant algorithm
  emulating an SW-safe solo-terminating MWSR register for $k$ clients
  with less than $2f+1$ servers.
\label{thm:app:2f+1}
\end{thm}

\begin{proof}
  Assume by contradiction that there exists an $f$-tolerant algorithm
  emulating an SW-safe solo-terminating MWSR register for $k>0$
  clients using $2f$ servers. By Corollary~\ref{cor:2f}, there exists
  a run $r_1$ of $A$ consisting of a single high-level write $W_1$ by
  a client $c_1$ such that
  $|{\cal S} \setminus \delta(Cov(t_0))| > 2f$ where $t_0=0$. Since no
  base registers are covered at $t_0$,
  $|{\cal S} \setminus \delta(Cov(t_0))| = |{\cal S}| > 2f$. However,
  by assumption, $|{\cal S}| = 2f$. A contradiction.
\end{proof}

\section{Correctness of CAS-ABD}
\label{sec:rwcorrectness}

We first argue that our emulation is a linearizable implementation of
ABDO. The argument relies on the following auxiliary invariants.

% We first show that $t \ge exp.ts$ every time line~\ref{line:cas} is
% reached.
\begin{lemma}
  If line~\ref{line:cas} is reached, then $t > exp.ts$.
\label{lem:exp}
\end{lemma}

\begin{proof}
  The proof is by induction on the number of iteration of the loop in
  lines~\ref{line:loop-start}--\ref{line:loop-end}. For the base case,
  note that line~\ref{line:test}, $t > exp.ts$ is the necessary
  condition for entering the loop. Hence, the lemma holds first time
  line~\ref{line:cas} is reached. Next, assume that the result is true
  for all iterations $k\ge 1$, and consider iteration $k+1$. % By
  % the induction hypothesis, $t \ge exp.ts$ at line~\ref{line:cas} at
  % iteration $k$.
  Since iteration $k+1$ is reached, the condition in
  line~\ref{line:aftercas} must be false at iteration $k$, that is,
  $old.ts < t$. By line~\ref{line:exp-old}, at the beginning of
  iteration $k+1$, $exp = old$, and therefore, $exp.ts = old.ts < t$
  as needed.
\end{proof}

\noindent
We now show that $b_i.ts$ is non-decreasing:

\begin{lemma}
  Let $b_i.ts_1$ and $b_i.ts_2$ be the values of $b_i.ts$ at times
  $t_1$ and $t_2$ respectively. If $t_1 < t_2$, then
  $b_i.ts_1 \le b_i.ts_2$.
\label{lem:non-decreasing}
\end{lemma}

\begin{proof}
  Observe that $b_i.ts$ can only change as a result of a successful
  \CAS\/ invocation in line~\ref{line:cas}. The necessary condition
  for that to happen is $exp = b_i$ in line~\ref{line:cas}. By
  Lemma~\ref{lem:exp}, $t > exp.ts = b_i.ts$. Hence, the value of
  $b_i.ts$ is either left unchanged, or increases as needed.
\end{proof}

\noindent
Next, we show linearizability:

\begin{lemma}
  Let $r$ be a run of the ABDO emulation in
  Algorithm~\ref{alg:cas-abd}, and $\sigma$ be a sequential run
  obtained from $r$ by shrinking each $update$ and $read$ operation to
  occur at its linearization point. Then, $\sigma$ is a sequential run
  of the ABD object in Algorithm~\ref{alg:updateABD}.
\label{lem:app:linear}
\end{lemma}

\begin{proof}
  Let $t_1,t_2,\dots$ such that $t_i < t_{i+1}$, $i\ge 1$, denote the
  times at which the linearization points occur in $r$.  The proof is
  by induction on $t_i$. For the base case, consider the first
  linearization point $t_1$. If $t_1$ is the linearization point of
  $read$, then its return value $((0, 0), v_0)$; and if $t_1$ is the
  linearization point of $update$, then its return value is
  $ack$. Since both return values are identical to those produced by
  the $read$ and $update$ of the ABD object if invoked at the initial
  state, the result holds.

  Next, assume that the result is true for the first $k-1$
  linearization points, and consider the $k$th linearization point
  $t_k$. If $t_k$ is the linearization point of $update$, then its
  return value is $ack$, which is consistent with the sequential
  specification of the ABD object. 

  Suppose that $t_k$ is the linearization point of a read
  operation. Suppose that the linearization point $t_{k-1}$ is
  associated with a read.  Since for any value of $exp$,
  $\CAS(exp ,exp)$ does not changes the content of $b_i$, the return
  value of read will be the same as that of the read linearized at
  $t_{k-1}$, which complies with the sequential specification of the
  ABD object.

  Next, suppose that the operation linearized at $t_{k-1}$ is an
  update operation $u = update(b_i, t, v)$ for some $t\in TS$ and
  $v \in \mathbb{V}$. Let $x_j$ denote the value of variable $x$ at
  time $t_j$.  The sequential specification of the ABD object requires
  the read to return $(t, v)$ if $t > b_i.ts_{k-2}$, and $b_{i,k-2}$,
  otherwise. We show that this is indeed the case.

  First, suppose that $t > b_i.ts_{k-2}$. Since no linearization
  points occur between $t_{k-2}$ and $t_{k-1}$, and $b_i$ can only be
  changed at a linearization point, at line~\ref{line:test},
  $exp.ts_{k-1} \le b_i.ts_{k-2} = b_i.ts_{k-1} < t$.  Hence,
  linearization point $t_{k-1}$ must occur at
  line~\ref{line:cas}. This means that \CAS\/ in line~\ref{line:cas}
  is successful as otherwise $old.ts_{k-1} \ge t$ implies that
  $old.ts_{k-1} = b_i.ts_{k-1} = b_i.ts_{k-2} \ge t$ contradicting the
  assumption. Therefore, the linearization point $t_{k-1}$ coincides
  with a successful \CAS\/ in line~\ref{line:cas} so that
  $b_{i,k-1} = (t, v)$.  Since no linearization points occur between
  $t_{k-1}$ and $t_k$, and $b_i$ can only be changed at a
  linearization point, $b_{i,k-1} = b_{i,k} = (t, v)$. Hence, the read
  will return $(t, v)$ as needed. 

  Finally, suppose that $t \le b_i.ts_{k-2}$. If $t \le exp.ts$, then
  linearization point $t_{k-1}$ occurs at line~\ref{line:test}, and
  therefore, $u$ returns without changing $b_i$. Hence,
  $b_{i,k-1} = b_{i,k-2}$. Suppose $t > exp.ts$, and consider the
  \CAS\/ invocation occurring at the first iteration of the loop in
  lines~\ref{line:loop-start}--\ref{line:loop-end}. Observe that this
  invocation must be unsuccessful as otherwise,
  $b_i.ts_{k-2} = exp.ts < t$ contradicting the assumption that
  $t \le b_i.ts_{k-2}$. At the same time,
  $old.ts = b_i.ts_{k-2} \ge t$. Hence, the condition in
  line~\ref{line:aftercas} is true, which implies that $u$ leaves the
  loop without changing the value of $b_{i,k-2}$ at $t_{k-1}$. We
  conclude that $b_{i,k-1} = b_{i,k-2}$. Thus, the read will return
  $b_{i,k-2}$ as required.
\end{proof}

\noindent
We next show that the ABDO emulation is wait-free if all its runs
satisfy timestamp uniqueness.

\begin{lemma}
  The ABDO emulation in Algorithm~\ref{alg:cas-abd} is wait-free
  provided all its runs satisfy timestamp uniqueness.
\label{lem:app:wait-free}
\end{lemma}

\begin{proof}
  Since the read operation is obviously wait-free, we only need to
  show that the update operation is wait-free as well.

  Consider an update invocation $u = update(b_i, t, v)$. If the
  condition in line~\ref{line:test} is false, then $u$ returns, and we
  are done. Otherwise, let $ts_j$, $j \ge 1$, be the value of $b_i.ts$
  before \CAS\/ is invoked at the $j$th iteration of the loop in
  lines~\ref{line:loop-start}--\ref{line:loop-end}. 

  At all iterations $j \ge 1$, if $ts_j \ge t$, then the condition in
  line~\ref{line:aftercas} is true, and the loop
  terminates. Otherwise, by Lemma~\ref{lem:non-decreasing} and
  timestamp-uniqueness, $ts_{j+1} > ts_j$. Since there are only
  finitely many timestamps between $ts_1$ and $t$, there exists an
  iteration where the condition in line~\ref{line:aftercas} is
  satisfied, and the loop terminates.
\end{proof}

\noindent
Given that timestamp uniqueness holds in all runs of MW-ABD, we
receive the following:

\begin{thm}
  The CAS-ABD algorithm is an $f$-tolerant implementation of a
  wait-free MWMR atomic register.
\label{thm:cas-abd}
\end{thm}

\begin{lemma} \label{la:tsgap} Let $op$ be an operation that invokes
  $update$ at time $t$ and let $op'$ be another operation that starts
  at time $t' \geq t$. If $k$ operations are invoked but do not
  complete before time $t$ then $ts(op').num \geq ts(op).num - k - 1$
\end{lemma}
\vspace*{-0.2cm}
\begin{proof}
  Let $op''$ be the operation with the highest timestamp that returns
  before time $t$. By MW-ABD timestamp selection mechanism,
  $ts(op') \geq ts(op'')$. Therefore it is sufficient to prove that
  $ts(op'').num \geq ts(op).num - k - 1$.  Suppose, for the purpose of
  contradiction, that $ts(op'').num = ts(op).num - k - 2$. Since every
  operation increments num by at most one and the timestamp of $op$ is
  $ts(op)$, at least $k + 1$ operations must be invoked before time
  $t$ with timestamps strictly greater than $ts(op).num - k - 2$. At
  least one of these operations returns before time $t$ by the
  statement of our Lemma. This is a contradiction since $op''$ was
  chosen to be the operation with the highest timestamp that returns
  before $t$.
\end{proof}
\vspace*{-0.3cm}

\begin{lemma} \label{la:concurrentobstructions} Let $op$ be an
  operation that invokes $update$ at time $t$ and $op'$ be another
  operation that obstructs $op$ on some object $b_i$ but is not one of
  the first two operations to obstruct $op$ on $b_i$. Then $op'$ does
  not complete by time $t$.
\end{lemma}
\vspace*{-0.2cm}
\begin{proof}
  Since $op$ is obstructed at least three times, the following
  sequence of invocations on $b_i$ must occur ($op.b_i.CAS$ denotes
  the invocation of $CAS$ on register $b_i$ during high-level
  operation $op$):\\ $op.b_i.CAS \ldots op.b_i.CAS \ldots op.b_i.CAS$.
  Since all three invocations of $op.b_i.CAS$ fail (the third one due
  to $op'$), we know that there are at least three invocations of
  $b_i.CAS$ by other operations that succeed:
  $op'''.b_i.CAS \ldots op.b_i.CAS \ldots op''.b_i.CAS \ldots
  op.v_i.CAS \ldots op'.b_i.CAS \ldots op.b_i.CAS$.

  Since $op'.b_i.CAS$ succeeds updating $b_i$, it learns the value
  written by $op''.b_i.CAS$, which happens after the first invocation
  of $op.b_i.CAS$, which in turn must occur after update is invoked
  during $op$, i.e., after time $t$. Hence, $op'$ does not complete by
  time $t$.
\end{proof}
\vspace*{-0.3cm}

\begin{lemma} \label{la:sametimestampcontention} Let $op$ be an
  operation. For any constant $n$ the number of operations $op'$ that
  are concurrent with $op$ and such that $ts(op').num = n$ is at most
  $\PntCont(op)$.
\end{lemma}
\vspace*{-0.2cm}
\begin{proof}
  Suppose for the sake of contradiction that there exists a constant
  $n$ such that there are $\PntCont(op)+1$ operations concurrent with
  $op$ with the first component of their timestamp equal to $n$. Since
  there are $\PntCont(op)+1$ operations and at most $\PntCont(op)$
  clients executing operations concurrently with $op$ at any single
  point in time (by definition of point contention), there is a
  client that executes two operations, both of which have the same
  first component of the timestamp. However, since each client
  executes operations sequentially, MW-ABD timestamp selection
  mechanism guarantees that the $num$ component of the first timestamp
  will be greater than that of the second one. A contradiction.
\end{proof}

\begin{thm}[Time Complexity]
  The CAS-ABD time complexity is {\em adaptive}\/ to concurrency
  guaranteeing that each operation $op$ terminates in at most $O(c^2)$
  base object accesses where $c=\PntCont(op)$.
\end{thm}
\vspace*{-0.2cm}
\begin{proof}
  Let $t$ be the time when $op$ invokes update. There are three types
  of operations that can obstruct $op$: (1) an operation that
  completes before time $t$; (2) an operation that starts but does not
  complete before time $t$; and (3) an operation invoked at time $t$
  or later. We next quantify the number of operations of each type
  that can obstruct $op$.

  By Lemma~\ref{la:concurrentobstructions} at most two operation
  completing before time $t$ can obstruct $op$ on a given
  register. Thus, at most two operations fall into the first
  category. By definition of $\PntCont(op)$, the number of operations
  of the second type is at most $\PntCont(op)$. By
  Lemma~\ref{la:tsgap}, this also implies that any operation $op'$ of
  the third type, that is, starting at time $t$ or later, satisfies
  $ts(op).num - ts(op').num \leq \PntCont(op) + 1$. Since operations
  with timestamps higher than $ts(op)$ cannot obstruct $op$ (see
  line~\ref{line:aftercas}), we only care about the case
  $0 \leq ts(op).num - ts(op').num$. There are at most
  $\PntCont(op) + 2$ numbers in this range. Since all operations that
  start at time $t$ or later and obstruct $op$ are concurrent with
  $op$, by Lemma~\ref{la:sametimestampcontention} there are at most
  $\PntCont(op)$ such operations whose first timestamp component is
  each of the numbers in the range described above. Overall, there are
  at most $(\PntCont(op)+2)*\PntCont(op)$ operations with timestamps
  in this range, and in total there are
  $\PntCont(op)^2 + 3\PntCont(op) + 2$ operations that may obstruct
  $op$.

  Notice that an operation $op'$ can obstruct $op$ on an object $b_i$
  only by changing the value of $b_i$ using $\CAS$ on
  line~\ref{line:cas}. By the specification of $\CAS$, the old value
  of $b_i$ was the expected value passed to $\CAS$ in this invocation
  during $op'$. By the conditions on lines~\ref{line:aftercas}
  and~\ref{line:loop-end}, once this $\CAS$ returns, update
  terminates, and $op'$ returns. This means that $op'$ can obstruct
  $op$ at most once. Since each operation can obstruct $op$ at most
  once, $\PntCont(op)^2 + 3\PntCont(op) + 2$ is an upper bound on the
  number of times a CAS invocation during $op$ can fail (for each
  object). % The interaction with different objects during $op$ is done
  % using different instances of update executed concurrently.
\end{proof}

\end{document}